\documentclass[11pt]{article}
\usepackage{latexsym, amssymb, amsmath, amsthm}

\usepackage[colorlinks=true,citecolor=blue,linkcolor=blue,urlcolor=blue,bookmarks,bookmarksopen,bookmarksdepth=2,backref=page]{hyperref}

\def \F {{\mathbb F}}

\def \Tr {{\rm Tr_n}}

\def \T {{\rm Tr}}

\def \+ {\oplus}
\newtheorem{theorem}{Theorem}[section]
\newtheorem{definition}[theorem]{Definition}
\newtheorem{lemma}[theorem]{Lemma}
\newtheorem{proposition}[theorem]{Proposition}
\newtheorem{remark}[theorem]{Remark}
\newtheorem{example}[theorem]{Example}
\newtheorem{corollary}[theorem]{Corollary}

\usepackage{longtable}

\newcommand{\Trn}{{\rm Tr}_n}

\begin{document}


\title{P$\wp$N functions, complete mappings and quasigroup difference sets}
\author{Nurdag\"{u}l Anbar$^{1}$, Tekg\"{u}l Kalayc\i$^{1}$, Wilfried Meidl$^{1,2}$, \\
Constanza Riera$^{3}$, Pantelimon St\u anic\u a$^4$	
	\vspace{0.4cm} \\
	\small $^1$Sabanc{\i} University,	
	\small MDBF, Orhanl\i, Tuzla, 34956 \. Istanbul, Turkey\\
	\small $^2$Institut f\"ur Mathematik, Alpen-Adria-Universit\"at Klagenfurt, Austria \\
    \small $^3$Department of Computer Science, 
    \small Electrical Engineering and Mathematical Sciences,\\
    \small Western Norway University of Applied Sciences, 5020 Bergen, Norway \\
    \small $^4$Department of Applied Mathematics, Naval Postgraduate School, \\
    \small Monterey, CA 93943--5216, USA \\
	\small Email: {\tt nurdagulanbar2@gmail.com}\\
	\small Email: {\tt tekgulkalayci@sabanciuniv.edu}\\
	\small Email: {\tt meidlwilfried@gmail.com} \\ 
    \small Email: {\tt csr@hvl.no} \\
    \small Email: {\tt pstanica@nps.edu} }

\date{}

\maketitle

\begin{abstract}
We investigate pairs of permutations $F,G$ of $\F_{p^n}$ such that $F(x+a)-G(x)$ is a permutation for every 
$a\in\F_{p^n}$. We show that necessarily $G(x) = \wp(F(x))$ for some complete mapping $-\wp$ of $\F_{p^n}$, and call
the permutation $F$ a perfect $\wp$ nonlinear (P$\wp$N) function. If $\wp(x) = cx$, then $F$ is a PcN function,
which have been considered in the literature, lately. With a binary operation on $\F_{p^n}\times\F_{p^n}$
involving $\wp$, we obtain a quasigroup, and show that the graph of a P$\wp$N function $F$ is a difference set 
in the respective quasigroup. We further point to variants of symmetric designs obtained from such quasigroup
difference sets. Finally, we analyze an equivalence (naturally defined via the automorphism group of the respective quasigroup) for P$\wp$N functions, respectively, the  difference sets in the corresponding quasigroup.  
\end{abstract}


\section{Introduction}

Let $F$ be a function from the finite field $\F_{p^n}$ with $p^n$ elements, into the finite field $\F_{p^m}$. Here, and throughout the paper, $m,n$ denote positive integers.
The derivative $D_aF$ of $F$ in direction $a\in \F_{p^n}$ is the function
\begin{equation*} 
D_aF(x) = F(x+a) - F(x).
\end{equation*}
For every $a\in\F_{p^n}$, $b\in\F_{p^m}$, we define
$\delta_F(a,b) = |\{x\in\F_{p^n}\,:\, D_aF(x) = b\}|$, and 
$\delta_F =  \max \{\delta_F(a,b)\,:\, a\in\F_{p^n}^* = \F_{p^n}\setminus\{0\}, b\in\F_{p^m}\}$. 
The function $F$ is called {\it differentially $\delta$-uniform}, if for all $a\in\F_{p^n}^*$ and $b\in\F_{p^m}$, 
we have $\delta_F(a,b)\le\delta$. The value $\delta_F$ is called the {\it differential uniformity} of $F$.

A low differential uniformity is crucial for functions from $\F_{2^n}$ to $\F_{2^m}$ used in
cryptography to thwart differential attacks~\cite{{Biham91}}. 

In the case $n=m$, and if and only if the characteristic $p$ is odd, there exist some functions $F$, called {\it planar functions}, for which every derivative $D_aF$, $a\ne 0$, is a permutation, i.e.,
$\delta_F = 1$. More general, a function $F:\F_{p^n}\rightarrow\F_{p^m}$
is called {\it perfect nonlinear (PN)} or a {\it bent function}, if the derivative $D_aF$ is balanced for every nonzero
$a\in\F_{p^n}$, i.e., for every $a\in\F_{p^n}^*$, $b\in\F_{p^m}$, the equation $D_aF(x) = b$ has exactly 
$p^{n-m}$ solutions. 
Whereas such functions exist in odd characteristic for all integers $n$ and $m\le n$, if $p=2$, then $n$ must be even and $m\le n/2$ (see~\cite{nyb}).

Perfect nonlinear functions $F:\F_{p^n}\rightarrow\F_{p^m}$ have rich connections to 
objects from several areas like 
coding theory, geometry and combinatorics. Bent functions correspond to relative difference sets,
divisible designs, planar functions yield projective planes, and if they are quadratic, commutative semifields. 
Boolean bent functions correspond to difference sets. For these reasons, perfect nonlinear functions 
have been intensively investigated over the last decades \cite{carmen}, not only for characteristic $2$ where 
substantial motivation comes from cryptography, but also for odd characteristic $p$ \cite{survey}.

  Given a function from $\F_{p^n}$ to $\F_{p^m}$, the concept of a $c$-differential 
\begin{equation}
\label{c-diff}
_cD_aF(x) = F(x+a) - cF(x), \quad c\in\F_{p^m}^*, 
\end{equation}
has been introduced in Ellingsen et al.~\cite{efrst}.

A function $F$ from $\F_{p^n}$ to $\F_{p^m}$ is called {\it perfect $c$-nonlinear (PcN)}
(also {\it $c$-differential bent}), if $_cD_aF$ is balanced for all $a\in\F_{p^n}$. (In the classical case, when $c=1$, $a=0$ is exempted.)

As motivation for this definition, in several articles it is pointed to possible variants
of the differential attack. However, there is not yet such an attack, since the principal idea of
the differential attack, cancelling the key in the sum of two outputs of the S-box, does not seem to work if $c\ne 1$. However, we make some observations in the last section, where we point to future research, which may help in circumventing that. Further research is required, surely, but as it turns out, there are instances in  higher order differentials when the round keys may disappear. 

Further, many interesting connections to various kinds of difference sets, projective planes,
etc., that we see for the conventional derivative ($c=1$), are not clear when $c\ne 1$. 
Hence, one objective of this article is to endow the $c$-differential with some further meaning, and relate the concept to combinatorial objects.
For this purpose, we first extend the concept of the $c$-differential, and consider, for functions $F:\F_{p^n}\rightarrow\F_{p^m}$, 
more general, differentials of the form 
\begin{equation}
\label{wp-diff}
_\wp D_aF(x) = F(x+a) - \wp(F(x)),
\end{equation}
for a permutation $\wp$ of $\F_{p^m}$. The $c$-differential in $(\ref{c-diff})$ corresponds then to the permutation $\wp(x) = cx$; in the classical case, $\wp$ is the identity.

We call a function $F$ from $\F_{p^n}$ to $\F_{p^m}$ a P$\wp$N function if $_\wp D_aF$ is balanced for every $a\in\F_{p^n}$. Note that it is a requirement that $m\leq n$, though we are most interested in the case of $n=m$. Observe that, given a permutation $F$, we 
can write any permutation $G$ as $G(x) = \wp(F(x))$ for some permutation $\wp$.
Therefore in Equation $(\ref{wp-diff})$ for $m = n$, we consider pairs of permutations 
$F,G$ (though not obvious now, Proposition~\ref{Fwpper} implies that $F$ must be a permutation if $_\wp D_aF(x)$ is a permutation for every $a$) of $\F_{p^n}$ such that 
\[ F(x+a)- G(x) \] 
is a permutation for every $a\in\F_{p^n}$. This may lead to some interesting questions
on permutation polynomials and, as we will see, complete mappings. Notably, as for perfect nonlinear functions, we can assign such permutations 
$F$ to some variants of difference sets, namely difference sets in corresponding {\it quasigroups}.

We remark that the definition of the $\wp$-differential is independent from the representation of the $n$-dimensional 
vector space over $\F_p$, which throughout this article will be identified with the (additive group of the) finite field~$\F_{p^n}$. \\
%
%
%
The article is organized as follows.

In Section \ref{basic}, we describe some properties of $\wp$ and $F$ for P$\wp$N functions~$F$. We show that if for 
two permutations $F,G$ of $\F_{p^n}$ for which $F(x+a)- G(x)$ is a permutation for every $a\in\F_{p^n}$, then 
$G(x) = \wp(F(x))$ for an orthomorphism $\wp$ (i.e., for a complete mapping $-\wp$). We confirm that such pairs of nonlinear permutations $F,G$ exist for orthomorphisms $\wp$ other than $\wp(x) = cx$. (As we will see, linear permutations $F$ are trivial 
examples of P$\wp$N functions for every orthomorphism $\wp$.)

In Section \ref{diffset} we describe the quasigroup difference sets, which we can assign to a P$\wp$N function. 
As example, we hence analyze the 
incidence structure obtained from the development of the (quasigroup) difference set which we get from a P$\wp$N function for a linear orthomorphism $\wp$. Note that $\wp(x) = cx$ form a special class of linear orthomorphisms. There are now several examples PcN functions in the literature, many of them for $c=-1$. The incidence structure has an interpretation as a generalization of a design obtained from a difference set in a group.

In Section \ref{equiv}, we analyze the equivalence for P$\wp$N functions. Naturally, the automorphism group of the quasigroup 
which comprises the corresponding difference set has to be considered. We will show that equivalence for P$\wp$N functions is included in EA-equivalence. 

In Section~\ref{concl}, we discuss the above-mentioned observation related to differential attacks, and point to some perspectives for future research.

\section{Properties of P$\wp$N functions}
\label{basic}

So far, research concentrated on P$\wp$N permutations $F$ of $\F_{p^n}$ where $\wp(x) = cx$, $c\ne 1$, in which case we call
$F$ a PcN function. 

Differently from PN functions, which only can exist for $p$ odd, and for which, except for the Coulter-Matthews functions, all known examples
are quadratic functions (and correspond to commutative semifields), quite some examples of PcN functions, $c\ne 1$, are known (though, not too many for even characteristic):
\begin{itemize}
\item[-] Every linearized permutation of $\F_{p^n}$ is PcN for every $c \in \F_{p^n}\setminus\{1\}$.
\item[-] Every quadratic permutation over $\F_{p^n}$ is PcN for every $c\in \F_p\setminus\{1\}$, see \cite{BC21}.
\item[-] Further examples are given in the list in Appendix~B.
\end{itemize}

We first extend a result in \cite{BC21} to functions from $\F_{p^n}$ to $\F_{p^m}$.
\begin{proposition}
Let $p$ be an odd prime, $c\ne 1$ in the prime field $\F_p$, and let $F : \F_{p^n} \rightarrow \F_{p^m}$ be a quadratic function, where $m|n$.
Then $F$ is balanced if and only if $F$ is P$c$N.
\end{proposition}
\begin{proof}
We first show that a quadratic function $F : \F_{p^n} \rightarrow \F_{p^m}$ can be written as $F(x)=\T_{m}^n(G(x))$ for a quadratic function 
$G$ on $\F_{p^n}.$ Let $y \in \F_{p^m}$ such that $F(x_y)=y$ for some $x_y \in \F_{p^n}$. Then there exists an element $z_{y} \in \F_{p^n}$ such that
$\T_{m}^n(z_{y})=y $. Consider the function $G $ on $\F_{p^n}$ such that $G(x_y)=z_y$. Then $F(x_y)=\T_{m}^n(G(x_y))=\T_{m}^{n}(z_y)=y$, which implies 
that $F(x)=\T_{m}^n(G(x))$. Since $F$ is a quadratic function and the relative trace function is linear, $G$ is quadratic. 

As $F(x)=\T_{m}^n(G(x))$, $F(x)$ is P$c$N if and only if $F(x + a) -cF(x)=\T_{m}^n(G(x + a) - c G(x))$ is balanced for every $a \in \F_{p^n}$. 
By \cite[Equation (1)]{BC21}, for a quadratic function $G$ on $\F_{p^n}$ and $c\ne 1$ in the prime field $\F_p$,
we have
\[ G(x + a) - cG(x) = \alpha G\left(x + \frac{a}{\alpha}\right)+ \beta, \]
where $\alpha = 1-c\in\F_p^*$, and $\beta=G(a)- \alpha G\left({\frac{a}{\alpha }}\right) \in \F_{p^n}$. Therefore, $F(x)$ is P$c$N if and only if 
$F(x + a) -cF(x)= \T_{m}^n(\alpha G(x + \frac{a}{\alpha})+ \beta ) $ is balanced, which clearly holds if and only if  $\T_{m}^n(G(x))=F(x)$ is balanced. 
\end{proof}

In the next proposition we obtain some first results on P$\wp$N functions $F$ from $\F_{p^n}$ to $\F_{p^m}$
with $\wp$ not necessarily of the form $\wp(x) = cx$.
\begin{proposition}
\label{Fwpper}
Let $F:\F_{p^n}\rightarrow\F_{p^m}$ be a P$\wp$N function, with $\wp(x)\neq x$. Then $F$ is balanced and $\wp(x)-x$ is a permutation of $\F_{p^m}$.
In particular, for a permutation $\wp$ of $\F_{p^n}$, a P$\wp$N function $F:\F_{p^n}\rightarrow\F_{p^n}$ and the function $\wp(x)-x$ 
are always permutations.
\end{proposition}
\begin{proof}
Since $F$ is P$\wp$N, the function $F(x+a)-\wp (F(x))$ is balanced for all $a\in \F_{p^n}$. In particular, in the case $a=0$, the function
	\begin{align*}
	F(x)-\wp (F(x))=(x-\wp (x))\circ F(x)
	\end{align*}
is balanced. Therefore, $x-\wp (x)$ is onto. This implies that $x-\wp (x)$ is a permutation of $\F_{p^m}$. Since $x-\wp (x)$ is a
permutation, $((x-\wp (x))\circ F(x))^{-1}(b)=F^{-1}(c)$, where  $c-\wp (c)=b$. Therefore, the balancedness of $(x-\wp (x))\circ F(x)$ implies that $F$ is balanced.
\end{proof}

By Proposition \ref{Fwpper}, $\wp(x)$ and $\wp(x)-x$ being permutations, is a necessary requirement for having P$\wp$N functions, if $\wp(x)\neq x$.
A permutation $\wp$ for which $\wp(x)-x$ is also a permutation   is called an {\it orthomorphism} (observe that, if $p=2$, such a permutation is called a {\it complete permutation polynomial}). 
Note that $\wp$ is an orthomorphism if and only if $-\wp$ is a {\it complete mapping}.

We can reformulate Proposition \ref{Fwpper} for permutations as follows. 
\begin{corollary}
Let $F,G$ be permutations of $\F_{p^n}$ such that $ F(x+a)-G(x)$ is a permutation for
every $a\in\F_{p^n}$. Then $G(x) = \wp(F(x))$ for some orthomorphism $\wp$ of $\F_{p^n}$.
\end{corollary}
	
\begin{remark}
$F$ is P$\wp$N for some orthomorphism $\wp$ if and only if $F$ is P$\wp_c$N for the 
orthomorphism $\wp_c(x) = \wp(x) + c$, $c\in\F_{p^n}$. Hence we will always suppose 
that $\wp(0) = 0$.	
\end{remark}
\begin{remark}
	In multivariate representation, a linear orthomorphism is an invertible matrix $A$, for which $A-I$ is also invertible.
\end{remark}

The converse of Proposition \ref{Fwpper} does not hold in general, but as we show in the next proposition, it does
hold for linear functions $F$. This confirms that for every orthomorphism $\wp$, we have P$\wp$N functions, namely
at least the linear ones. 
\begin{proposition}
Let $\wp$ be an orthomorphism of $\F_{p^m}$, i.e., $\wp(x)$ and $\wp(x)-x$ are permutations of $\F_{p^m}$.
A linear function $F:\F_{p^n}\rightarrow\F_{p^m}$ is P$\wp$N if and only if $F$ is balanced.
In particular if $n=m$, then a linear function $F$ on $\F_{p^n}$ is P$\wp$N if and only if $F$ is a permutation. 
\end{proposition}
\begin{proof}
By Proposition \ref{Fwpper}, it is enough to show that if $F$ is balanced,
then $F$ is P$\wp$N. As $F$ is linear, for $a\in\F_{p^n} $ we have
		\begin{align*}
		F(x+a)-\wp (F(x))=F(x)+F(a)-\wp (F(x))=(x-\wp(x))\circ F(x)+F(a).
		\end{align*}
		Hence, $F(x+a)-\wp (F(x))$ is balanced if and only if $(x-\wp(x))\circ F(x)$ is balanced. As $x-\wp(x)$ is a permutation and $F(x)$ is balanced,
		we obtain the assertion.
\end{proof}
This result can be extended to affine functions $F$, if $\wp$ is linearized. 
%
\begin{proposition}
Let $\wp$ be a linearized orthomorphism of $\F_{p^m}$.
An affine function $F:\F_{p^n}\rightarrow\F_{p^m}$ is P$\wp$N if and only if $F$ is balanced.
In particular if $n=m$, then an affine function $F$ on $\F_{p^n}$ is P$\wp$N if and only if $F$ is a permutation. 
\end{proposition}
\begin{proof}
Again, by Proposition \ref{Fwpper}, it is enough to show that if $F$ is balanced,
then $F$ is P$\wp$N. We can write $F(x)=L(x)-\alpha$, where $L$ is a linearized function, and $\alpha\in\F_{p^m}$. As $L$ is linear, for $a\in\F_{p^n} $ we have
		\begin{align*}
		F(x+a)-\wp (F(x))&=L(x)+L(a)-\alpha-(\wp (L(x))-\wp(\alpha))\\
&=(x-\wp(x))\circ L(x)+L(a)-\alpha+\wp(\alpha).
		\end{align*}
		Hence, $F(x+a)-\wp (F(x))$ is balanced if and only if $(x-\wp(x))\circ L(x)$ is balanced. As $x-\wp(x)$ is a permutation and $L(x)$ is balanced (note that $L$ is balanced if and only if $F$ is balanced),
		we obtain the assertion.
\end{proof}

As every linear permutation is P$\wp$N for every orthomorphism $\wp$, and every
quadratic permutation is P$\wp$N for $\wp(x) = cx$, $c\in \F_p\setminus\{0,1\}$, one
may conclude that P$\wp$N functions are not very rare objects (for odd characteristic). In the meantime also 
several examples of nonlinear and not quadratic PcN functions are known for several 
values of $c$, see the tables in Appendix~B.

However, if we only slightly change the orthomorphism $cx$ to another linearized
monomial $\wp(x) = cx^{p^j}$, the situation becomes quite different.
With examples for small field size, one can computationally confirm that, in general, the 
simplest quadratic permutation of $\F_{2^n}$, the Gold function $F(x) = x^{2^k+1}$,
$\gcd(2^n-1,2^k+1) = 1$, is not P$\wp$N. In the following proposition it is shown that for sufficiently large finite fields, the function $x^{2^k+1}$ never is P$\wp$N for $\wp(x) = cx^{2^j}$. 
This supports the assumption that for orthomorphisms other than 
$\wp(x) = cx$, P$\wp$N functions are harder to find. 
\begin{proposition}
\label{NixPcN}
For a divisor $j$ of $n$, let $\wp(x) = cx^{2^j}$, where $c$ is not a $(2^j-1)$-th power in $\F_{2^n}$.
For every $j,k$, the Gold function $F(x) = x^{2^k+1}$ is not P$\wp$N for sufficiently large $n$ (with $\gcd(2^n-1,2^k+1) = 1$).
\end{proposition}
\begin{proof}
First note that $\wp$ is an orthomorphism since $c$ is not a $(2^j-1)$-th power (which always exists as $j$ divides $n$).
We can suppose that $\gcd(2^n-1,2^k+1) = 1$ since otherwise $x^{2^k+1}$ is not a permutation.

To show that $F$ is not P$\wp$N for all sufficiently large $n$, observe that
\begin{align*}
F(x+a)+\wp (F(x))&=(x+a)^{2^k+1}+cx^{2^j(2^k+1)}\\
&=cx^{2^j(2^k+1)}+x^{2^k+1}+ ax^{2^k}+a^{2^k}x+a^{2^k+1}.
\end{align*} 
That is, $F(x+a)+\wp (F(x))$ is a permutation of $\F_{2^n}$ if and only if 
$$H(x)=F(x+a)+\wp (F(x))+a^{2^k+1}=cx^{2^j(2^k+1)}+x^{2^k+1}+ ax^{2^k}+a^{2^k}x$$ is a permutation of $\F_{2^n}$. We will show that for sufficiently large $n$, there exists $a\in \F_{2^n}^*$ such that $H(x)$ has a nonzero root in $\F_{2^n}$, which gives the desired conclusion.\\
Set $x=ay\neq0$. Then we have
\begin{align*}
H(ay)=ca^{2^j(2^k+1)}y^{2^j(2^k+1)}+a^{2^k+1}y^{2^k+1}+ a^{2^k+1}y^{2^k}+a^{2^k+1}y.
\end{align*}
Hence $H(ay)=0$ if and only if 
\begin{align}\label{eq:a}
a^{(2^j-1)(2^k+1)}=c^{-1}\frac{y^{2^k+1}+ y^{2^k}+y}{y^{2^j(2^k+1)}}.
\end{align}
Recall that $a^{2^k+1}$ is a permutation of $\F_{2^n}$ and $y\neq 0$. Hence, by setting $z=a^{2^k+1}$, 
Equation \eqref{eq:a} holds if and only if 
\begin{align}\label{eq:z}
z^{2^j-1}=c^{-1}\frac{y^{2^k}+ y^{2^k-1}+1}{y^{2^j(2^k+1)-1}}.
\end{align}
Let $F$ be the function field of the curve defined by Equation \eqref{eq:z}, i.e., $F=\F_{2^n}(y,z)$ with 
$z^{2^j-1}=c^{-1}({y^{2^k}+ y^{2^k-1}+1})/{y^{2^j(2^k+1)-1}}$. Since $2^j-1$ is a divisor of $2^n-1$, $F/\F_{2^n}(y)$ is a Kummer extension of degree $2^j-1$. For the properties of Kummer extensions, we refer to \cite[Proposition 3.7.3]{sti}. Note that $p(T)=T^{2^k}+ T^{2^k-1}+1$ is a separable polynomial. Therefore, any zero of $y^{2^k}+ y^{2^k-1}+1$ is totally ramified in $F/\F_{2^n}(y)$. This implies that $\F_{2^n}$ is the full constant field of $F$. Hence, Equation \eqref{eq:z} defines an absolutely irreducible curve $\mathcal{X}$ over $\F_{2^n}$, see \cite[Corollary 3.6.8]{sti}. Then the Hasse-Weil bound 
\cite[Theorem 5.2.3]{sti} implies that $\mathcal{X}$ has a sufficiently large number of rational points for all sufficiently large values of $n$. Together with Bezout's theorem, we conclude that there exists a rational point $(y,z)\in \mathcal{X}$ satisfying $yz\neq 0$. As $z=a^{2^k+1}$ and $y=ax$, this shows the existence  of $a\in \F_{2^n}^*$ for which $G(x)$ has a nonzero root in $\F_{2^n}$. 
\end{proof}

\begin{remark}
One can show that Proposition \textup{\ref{NixPcN}} applies at least for $n > 2(j+k)$.
For the proof we refer to Appendix~A.
\end{remark}

We finish this section with some examples of  P$\wp$N functions for orthomorphisms other than $cx$, in both even and odd characteristics. We start with the  quadratic orthomorphism
$\wp(x) = bx + (x^{2^m}+x)^3$ on $\F_{2^n}$, $n=2m$.
\begin{theorem} \label{thm:ex}
Let $n = 2m$, $m$ odd, $b\in \F_{2^m} \setminus \{0,1\}$ and let $F(x)=bx+(x^{2^m} + x)^3$. Then $F$ is P$\wp$N for the orthomorphism
$\wp(x) = F(x) =bx+(x^{2^m} + x)^3$.
\end{theorem}
\begin{proof}
 We first show that $F$ is a permutation polynomial of $\F_{2^n}$. Let $\zeta$ be an element satisfying
$\zeta^2 + \zeta + 1=0$. Then $\{1, \zeta\}$ forms a basis of $\F_{2^n}$ over $\F_{2^m}$. 
Note that $\F_{4}=\F_{2}(\zeta)$, i.e., $\zeta^3=1$ and $\zeta^{2^m}=(\zeta^{2^{m - 1} - 1}\zeta)^2=\zeta^2=\zeta + 1$, as 
	$m$ is odd, and hence $3 \ | \ 2^{m - 1} - 1$. 
Writing $x \in \F_{2^n}$ as $x = y + \zeta z$ for some unique $y, z \in \F_{2^m}$, we have the following equalities.
	\begin{align*}  
	F(x) = bx+(x^{2^m} + x)^3
	&=b( y + \zeta z) +( (y + \zeta z)^{2^m} + y + \zeta z)^3
	\\&=b y + \zeta b z + (\zeta^{2^m} + \zeta)^3z^3
	\\&=by + z^3 + \zeta bz,
	\end{align*}
where we used $\zeta^{2^m}+\zeta+1=0$ in the second equality. 	
We immediately see that $F$ is onto, hence a permutation. In fact for given $u + v\zeta$, we have 
$F(y+z\zeta) = u+v\zeta$ with $z = v/b$ and $y = (u+z^3)/b$ ($z = v/b$).
As readily seen, $F$ is an orthomorphism as $\bar{b}x+(x^{2^m} + x)^3$ is a permutation for $\bar{b} = b+1 \not\in\{0,1\}$.
%
	
It remains to show that 
	$F(x + a) + F(F(x))$ is a permutation polynomial of $\F_{2^n}$ for all $a \in \F_{2^n}$. We have
	\begin{align*}
	&F(x + a) + F(F(x))\\
	&=	b(x + a) + (x^{2^m} + x + a^{2^m} + a)^3 + F(bx + (x^{2^m} + x)^3)\\
	& =b(x + a) + (x^{2^m} + x)^3 + (x^{2^m} + x)^2(a^{2^m} + a) \\
	&\qquad 	+ (x^{2^m} + x)(a^{2^m} + a)^2  + (a^{2^m} + a)^3    + b(bx + (x^{2^m} + x)^3) \\
	&\qquad + \left((bx + (x^{2^m} + x)^3)^{2^m} + (bx + (x^{2^m} + x)^3)\right)^3 \\
	&  =bx + ba + (x^{2^m} + x)^3 + (x^{2^m} + x)^2(a^{2^m} + a) \\
	&\qquad + (x^{2^m} + x)(a^{2^m} + a)^2 + (a^{2^m} + a)^3  + b^2x + b(x^{2^m} + x)^3\\
	&\qquad + \left( b^{2^m}x^{2^m} + (x^{2^m} + x)^3 + bx + (x^{2^m} + x)^3\right)^3.
	\end{align*}
We can ignore the constant term $ba + (a^{2^m} + a)^3$, set $a^{2^m} + a = c \in\F_{2^m}$, and show that 
\begin{align*}
H(x)=(b + b^2)x + (b + 1)(x^{2^m} + x)^3 + c(x^{2^m} + x)^2 + (c^2 + b)(x^{2^m} + x)
\end{align*}
is a permutation polynomial, where $H(x)=F(x + a) + F(F(x))+ba + (a^{2^m} + a)^3$.
 Substituting  $x=y + \zeta z$ into $H(x)$, and using the fact
that $x^{2^m} + x = (y + \zeta z)^{2^m} + y + \zeta z = z$, we obtain
\begin{align*}
H(y + \zeta z)&=(b + b^2)y + (b+1)z^3 + cz^2 + (c^2 + b)z + \zeta(b + b^2)z.
\end{align*}
Given $u,v\in\F_{2^m}$, we uniquely obtain $z = v/(b+b^2)$ and then $y$ from $H(y+z \zeta) = u + v\zeta$. Note that $b+b^2\ne 0$.
Consequently, $H$ is onto, hence a permutation.
\end{proof}
%
%
%

\begin{remark}
We observe that $F(x)$ given in Theorem~\textup{\ref{thm:ex}} stays P$\wp$N for any extension $\F_{2^{kn}}$ of $\F_{2^n}$ of odd degree~$k$. In other words, $F(x)$ stays P$\wp$N for infinitely many extensions $\F_{2^n}$. We call such an $F(x)$, an exceptional P$\wp$N function.
\end{remark}

We now provide a class of P$\wp$N functions on $\F_{p^n}$ for odd characteristic~$p$. 
\begin{theorem}
Let $m,k$ be  positive integers, $n=2m$, and $q$ a power of an odd prime~$p$. Let   $F(x)=\left(x^{q^m}-x\right)^{2k}- x\in\F_{q^n}[x]$.
Then $F$ is a P$\wp$N function with respect to the orthomorphism~$\wp(x)=F(x)$.
\end{theorem}
\begin{proof}
In~\cite{YD14} it was shown (with our notations) that if $\delta^{q^m}=-\delta$ and  $L$ is a $p$-linearized polynomial, then a polynomial of the form $G(x)=(x^{q^m}-x+\delta)^{2k}+L(x)$  is a permutation polynomial on $\F_{q^n}$ if and only if $L$ is a permutation polynomial.  Taking $L(x)=-x$ and $\delta=0$, it follows that $F=\wp$ is a permutation polynomial.

We next compute
\begin{align*}
F(F(x))&= \left( (F(x)^{q^m}-F(x)\right)^{2k}-F(x)\\
&= \left( \left(\left(x^{q^m}-x \right)^{2k}- x \right)^{q^m}-\left(x^{q^m}-x\right)^{2k}+ x \right)^{2k}  -\left(x^{q^m}-x \right)^{2k}+x\\
&=\left(  \left(x^{q^m}-x \right)^{2k q^m}- x^{q^m} -\left(x^{q^m}-x \right)^{2k}+ x\right)^{2k} -\left(x^{q^m}-x \right)^{2k}+x\\
&= \left(  \left(x^{q^{2m}}-x^{q^m} \right)^{2k}- x^{q^m} -\left(x^{q^m}-x\right)^{2k}+ x\right)^{2k}  -\left(x^{q^m}-x \right)^{2k}+x\\
&= \left(  \left(x-x^{q^m}  \right)^{2k}- x^{q^m} -\left(x^{q^m}-x \right)^{2k}+ x\right)^{2k} -\left(x^{q^m}-x \right)^{2k}+x\\
&=x,
\end{align*}
that is, $F$ is self-invertible.
  Further, we write
\begin{align*}
F(x+a)-\wp(F(x))&= \left(x^{q^m}-x -a^{q^m}+a \right)^{2k}-(x+a) -x.
\end{align*}
It will be sufficient to show that   $H_a(x)=\left(x^{q^m}-x -a^{q^m}+a \right)^{2k}-2x$
is a permutation. Taking $\delta=a^{q^m}-a$, $L(x)=-2x$ and observing that $\delta^{q^m}=-\delta$, for all $a$, using again the result of Yuan and Ding~\cite{YD14}, we infer that $H_a$ is a permutation.
\end{proof}

\section{P$\wp$N functions, difference sets, and designs}
\label{diffset}

Let $G$ be a group of order $\mu\nu$ with a normal subgroup $N$ of order $\nu$. A $k$-subset $D$ of $G$ 
is called a {\it $(\mu,\nu,k,\lambda)$ relative difference set, relative to $N$}, if every element of $G\setminus N$ 
can be written as a difference of two elements in $D$ in exactly $\lambda$ ways, and there is no representation 
of this form for any nonzero element in $N$. A relative difference set is a generalization of a {\it $(v,k,\lambda)$ 
difference set}, which we can see as a relative difference set with trivial $N = \{0\}$ ($\nu = 1$, $v = \mu\nu = \mu$).

The set of all translates $\{D+a\,:\,a\in G\}$ of a subset $D$ in a group $G$, is called the {\it development}
of $D$. If $D$ is a $(\mu,\nu,k,\lambda)$ relative difference set, then the development of $D$ gives rise to
a {\it divisible design} (the points are the elements of the group, the blocks are the translates $D+a$).
We refer to \cite{jung82} for further information on divisible designs.
If $D$ is a difference set ($\nu = 1$), then we obtain a symmetric design, where every two distinct points
are simultaneously on exactly $\lambda$ blocks, and every two blocks intersect in exactly $\lambda$ points.

The definition of a bent function via balanced conventional derivatives has an equivalent version in terms of relative
difference sets. A function $F:\F_{p^n}\rightarrow\F_{p^m}$ is bent if and only if the graph of $F$,
$\mathcal{G}_F = \{(x,F(x))\,:\,x\in\F_{p^n}\}$ is a $(p^n,p^m,p^n,p^{n-m})$ relative difference set in 
$\F_{p^n}\times\F_{p^m}$ relative to $\{0\}\times\F_{p^m}$.

Bent functions hence give also rise to divisible designs. We remark that for a planar function, the divisible
design can be transformed into a projective plane, see \cite[Section 3.3]{pott16}.

The objective of this section is to relate P$\wp$N functions with the quasigroup difference sets. We then analyze the incidence structure obtained from the development of the (quasigroup) difference set,
which we get from a P$\wp$N function for a linear orthomorphism. The developments of these sets exhibit then properties comparable to those of designs.

\subsection{P$\wp$N functions and quasigroup difference sets}

For a permutation $\wp$ of $\F_{p^m}$, we define the binary operation $+_\wp$ on the set $\F_{p^n}\times\F_{p^m}$ as
\[ (x_1,y_1) +_\wp (x_2,y_2) = (x_1+x_2,y_1 + \wp(y_2)). \]
Recall that a set $Q$ with a binary operation $\star$ is called a {\it quasigroup}, if the equations $a\star x = b$ and $y\star a = b$ 
have a unique solution for all $a,b$ in $Q$. Consequently, $(\F_{p^n}\times\F_{p^m},+_\wp)$ becomes a quasigroup. 

We extend the definition of a difference set in finite groups to finite quasigroups. Note that for calculating the set of all differences
in a subset $D$ of $(Q,\star)$, it is only required that  $(Q,\star)$ is a quasigroup. 
\begin{definition}
	Let $D$ be a $k$-subset of a quasigroup $Q$ of order $v$. Then $D$ is called a $(v,k,\lambda)$ 
	quasigroup difference set in $Q$, if every element of $Q$ can be written as a difference of two elements in $D$ 
	in exactly $\lambda$ ways.
\end{definition}
\begin{remark}
	There is a slight difference between the definition of a difference set in a group and in a quasigroup.
	For a difference set $D$ in a group $G$, all but the zero element (which not necessarily exists in a
	quasigroup) can be written as a difference of elements in $D$ in $\lambda$ ways. Clearly, in a group,
	the element $0$ can be written as a difference of two elements of $D$ in exactly $|D| = k$ ways. 
	As a consequence, whereas the parameters of a difference set in a group satisfy $k(k-1) = (v-1)\lambda$,
	for a quasigroup difference set we have $k^2 = v\lambda$.
\end{remark}
We can use any permutation~$\wp$ to define a quasigroup, as given above. Our objective is to relate P$\wp$N functions
to difference sets in such a quasigroup. By this we intend to point to some similarities to perfect nonlinear functions.
In the light of Proposition~\ref{Fwpper}, we restrict ourselves to $\wp$ being an orthomorphism.
%
\begin{theorem}
	Let $F:\F_{p^n}\rightarrow\F_{p^m}$ be a P$\wp$N function. Then the graph $\mathcal{G}_F$ 
	of $F$ is a $(p^{n+m},p^n,p^{n-m})$ difference set in $(\F_{p^n}\times\F_{p^m},+_\wp)$.
	In particular, if $m=n$, then $\mathcal{G}_F$ is a $(p^{2n},p^n,1)$ difference set in 
	$(\F_{p^n}\times\F_{p^n},+_\wp)$.
\end{theorem}	
\begin{proof} First note that in a quasigroup $(Q,\star)$, $a$ is the difference of $b$ and $c$, if $c$ is the unique
solution of $a\star x = b$. In our quasigroup, $(a,b)$ is then the difference of $(x_1,y_1)$ and $(x_2,y_2)$
if $a+x_2 = x_1$ and $b + \wp(y_2) = y_1$, i.e., $a = x_1-x_2$ and $b = y_1-\wp(y_2)$.

For a fixed $(a,b)\in\F_{p^n}\times\F_{p^m}$, we determine the number of possibilities to write 
$(a,b)$ as a difference (in the quasigroup) of two distinct elements in $\mathcal{G}_F$, i.e., 
$(a,b) = (x_1-x_2, F(x_1)-\wp(F(x_2)))$. This is exactly the number of $x$, such that $(a,b) = (a,F(x+a)-\wp(F(x)))$.
Since $F$ is a P$\wp$N function, this number is exactly $p^{n-m}$ (for all $a\in\F_{p^n}$ and $b\in\F_{p^m}$).
\end{proof}
%
%
%
%

%

\subsection{PcN functions and designs}

Whereas the difference set is defined with the group operation, for the design (which is simply an incidence
relation between blocks and points) constructed from the difference set, the structure of the group is not relevant.
We attempt to assign a class of incidence structure to the quasigroup difference sets obtained with P$\wp$N functions. 
As one may expect, the independence of the number of solutions for $_\wp D_aF(x) = b$ from $a$ 
and $b$, again transfers to a property (in terms of $\lambda$) that pairs of points of the incidence structure satisfy
simultaneously. For simplicity, we consider the case of linearized  orthomorphisms $\wp(x)$. There are several examples for such functions known in the case $\wp(x) = -x$, see the list in Appendix~B. 


The development of the graph of $F$ in the quasigroup $(\F_{p^n}\times\F_{p^m}, +_\wp)$ consists of the
sets $(x,F(x)) +_\wp (u,v) = (x+u,F(x) + \wp (v))$,  $x\in \F_{p^n}$. Since $\wp (v)$ runs through $\F_{p^m}$ if $v$ does,
this equals the set of the conventional translates $(x,F(x)) + (u,v)$.
\begin{theorem}
\label{designThm}
Let  
$F:\F_{p^n}\rightarrow\F_{p^m}$ be a P$\wp$N function, where $\wp$ is linear. 
Then the development of the graph of $F$ yields an incidence structure with $v = p^n\times p^m$
points and $v$ blocks, with the following properties$:$
\begin{itemize}
\item[(i)] Every block contains $k = p^n$ points (elements of $\F_{p^n}\times\F_{p^m}$), and every 
point is on exactly $k$ blocks.
\item[(ii)] The block set separates into a class of $r = p^n$ single blocks $B_1,\ldots,B_r$ and a class 
of $2s = (p^m-1)p^n$ paired blocks $B_{r+1},\bar{B}_{r+1}, \ldots, B_{r+s},\bar{B}_{r+s}$.
For every pair of points, the number of blocks among $B_1, \ldots, B_r$ which contain both elements 
simultaneously, and the number of pairs of blocks among $B_{r+1},\bar{B}_{r+1}, \ldots, 
B_{r+s},\bar{B}_{r+s}$, for which one element of the pair of points is in $B_{r+i}$ and the other one is in 
$\bar{B}_{r+i}$, always add to $ p^{n-m}$. 
\item[(iii)] The point set separates into a class of $r$ single points $P_1,\ldots,P_r$ and a class 
of $2s$ paired points $P_{r+1},\bar{P}_{r+1}, \ldots, P_{r+s},\bar{P}_{r+s}$.
For every pair of blocks, the number of points among $P_1, \ldots, P_r$, which are on both blocks, and the 
number of pairs of points among $P_{r+1},\bar{P}_{r+1}, \ldots, P_{r+s},\bar{P}_{r+s}$, for which 
one point of the pair is in one of the two blocks, the other one is in the other, always add to $p^{n-m}$. 
\end{itemize}
\end{theorem}
\begin{proof}
We first show  $(i)$. By definition, every block has $k$ points. For a given point $(x_1,y_1) \in \F_{p^n}\times\F_{p^m}$, we pick an
element $(x,F(x))$ of the graph of~$F$. Then $(x_1,y_1) = (x,F(x)) + (u,v)$ for a unique 
$(u,v) \in \F_{p^n}\times\F_{p^m}$, and $(x_1,y_1)$ is a point of $\mathcal{G}_F + (u,v)$. Since there are 
$k$ choices for an element of the graph, $(x_1,y_1)$ is on $k$ blocks.

Next, we look at $(ii)$.
We now divide the blocks into two classes. The first class consists of the $p^n$ blocks of the form
$\mathcal{G}_F + (u,0)$, $u\in\F_{p^n}$. In the second class, the set of the remaining blocks, we form 
$p^n(p^m-1)/2$ pairs of the form $\mathcal{G}_F + (u,v)$, $ \mathcal{G}_F + (u,\wp^{-1}(v))$. Note that $\wp^{-1}(v)=v$ if and only if $\wp(v)=v$, i.e., $-\wp(v)+v=0$. Since  $-\wp(v)+v$ is a permutation of $\F_{p^m}$, this holds if and only if $v=0$. Hence, $(u,v)=(u,\wp^{-1}(v))$ happens if and only if  $(u,v)=(u,0)$. 

Suppose that $(x_1,y_1)\in \mathcal{G}_F + (u,v)$ and $(x_2,y_2)\in \mathcal{G}_F +(u,\wp^{-1}(v))$. That is, $(x_1,y_1)  = (d_1,F(d_1)) + (u,v) $ and $(x_2,y_2)  = (d_2,F(d_2))+(u,\wp^{-1}(v)) $ for some $d_1,d_2 \in \F_{p^n}$. In other words, by the linearity of $\wp$ we have
\begin{align}\label{eq:uv}
u&=x_1-d_1=x_2-d_2, \; \text{and} \\ \nonumber
 v&=y_1-F(d_1)=\wp(y_2)-\wp(F(d_2)).
\end{align}
Then by setting $a=d_1-d_2$, we have
\begin{align*}
(x_1,y_1)-_\wp (x_2,y_2)&=(x_1-x_2, y_1-\wp(y_2))=(a, F(a+d_2)-\wp(F(d_2))).
\end{align*} 
Since $F$ is P$\wp$N, for $a\in \F_{p^n}$ the equation $F(a+d_2)-\wp(F(d_2))=b$ has $p^{n-m}$ solutions. In particular, the number $d_2$, equivalently, the number of $u$, satisfying Equation \eqref{eq:uv} is $p^{n-m}$. As we remarked, if $v=0$, then the pairs lie in the same block $\mathcal{G}_F + (u,0)$; otherwise they lie in the distinct blocks, namely, $\mathcal{G}_F + (u,v)$, $ \mathcal{G}_F + (u,\wp^{-1}(v))$. 

Finally, we show $(iii)$.
The class of single points is given by the points $(x,0)$, $x\in\F_{p^n}$. The remaining points form $s$ pairs $(x,y)$, $(x,\wp^{-1}(y))$. Similarly, note that $\wp^{-1}(y)=y$ if and only if $y=0$, i.e., for nonzero $y\in \F_{p^n} $, the pairs  $(x,y)$, $(x,\wp^{-1}(y))$ are distinct.

Fix arbitrarily two (distinct) blocks $\mathcal{G}_F + (u_1,v_1)$, $\mathcal{G}_F + (u_2,v_2)$. We are interested in 
the number of $(x,y)\in\F_{p^n}\times\F_{p^m}$, such that $(x,y)$ is in $\mathcal{G}_F + (u_1,v_1)$ and $(x,\wp^{-1}(y))$
is in $\mathcal{G}_F + (u_2,v_2)$. By the linearity of $\wp$, this is equivalent to
\begin{align}\label{eq:xy}
x & = d_1 + u_1 = d_2 + u_2, \;\mbox{and} \\ \nonumber
y & = F(d_1) + v_1 = \wp(F(d_2)) +\wp( v_2)
\end{align}
for some $d_1,d_2\in\F_{p^n}$. Set $a=u_2-u_1$ and $b=\wp( v_2)-v_1$. Hence, we are looking for the number of elements $d_2$ satisfying $F(d_2+a) - \wp(F(d_2))=b$. Since $F$ is P$\wp$N, the equation has $p^{n-m}$ solutions, i.e., there are $p^{n-m}$ elements $d_2$ satisfying the equation. In other words, the number of elements $x\in\F_{p^n}$ satisfying Equation~\eqref{eq:xy} is $p^{n-m}$. If $y\neq 0$, i.e., $F(d_2) + v_2\neq 0$, then the pairs $(x,y)$, $(x,\wp^{-1}(y))$ lie in the distinct blocks $\mathcal{G}_F + (u_1,v_1)$, $\mathcal{G}_F + (u_2,v_2)$; otherwise $(x,0)$ lies in both.
\end{proof}
\begin{remark}
As for the conventional symmetric designs, by the properties (i), (ii), (iii) in Theorem~\textup{\ref{designThm}}, the dual of the 
incidence structure, i.e., the incidence structure obtained by changing the roles of points and blocks, has the 
same properties.
\end{remark}

\begin{remark}
A design (which we obtain with the development of a difference set in a group) is an incidence structure as
given in Theorem~\textup{\ref{designThm}}, for which all blocks (points) belong to the first class.
\end{remark}

\section{Equivalence for P$\wp$N functions}
\label{equiv}

Clearly, the $P\wp$N property is not invariant under the classical (extended) affine equivalence, using automorphisms of the
elementary abelian group. In order to define equivalence between P$\wp$N functions, one apparently has to look at the automorphism 
group corresponding to the quasigroup operation. 
Two functions $F_1, F_2: \F_{p^n}\rightarrow\F_{p^m}$ are  {\it $\wp$-affine equivalent}, if $F_2 = \mathcal{A}_2(F_1(\mathcal{A}_1(x)))$, where 
$\mathcal{A}_i=\mathcal{L}_i+\alpha_i$, $\alpha_1\in\F_{p^n}$, $\mathcal{L}_1$ is a linearized permutation of 
$\F_{p^n}$ (as for conventional affine equivalence), $\alpha_2\in\F_{p^m}$ and $\mathcal{L}_2$ is a permutation of $\F_{p^m}$ such that
$\mathcal{L}_2(y_1 + \wp(y_2)) = \mathcal{L}_2(y_1) + \wp(\mathcal{L}_2(y_2))$ for all $y_1,y_2\in\F_{p^m}$.
Note that, more precisely, $\mathcal{A}_2=\mathcal{L}_2+\wp(a)$ for some $a\in\F_{p^m}$. Using that $\wp$ is a permutation, with $\wp(a) = \alpha_2$,
we can write $\mathcal{A}_2$ as above.
The graphs of $F_1$, $F_2$ are then equivalent as subsets of the quasigroup $(\F_{p^n}\times\F_{p^m},+_\wp)$.

If $\mathcal{A}_1,\,\mathcal{A}_2$ are linearized, we call this equivalence the {\it $\wp$-linear equivalence}.
\begin{lemma}
\label{equilem}
Let $\wp$ be a permutation of $\F_{p^m}$ (an orthomorphism) with $\wp(0) = 0$, and let $\mathcal{L}$ be a permutation of $\F_{p^m}$,
such that $\mathcal{L}(y_1 + \wp(y_2)) = \mathcal{L}(y_1) + \wp(\mathcal{L}(y_2))$ for all $y_1,y_2\in\F_{p^m}$. Then $\mathcal{L}$ 
is a linear permutation.
\end{lemma}
\begin{proof}
Let $\mathcal{L}(y_1 + \wp(y_2)) = \mathcal{L}(y_1) + \wp(\mathcal{L}(y_2))$. For the unique
$y_2$ for which $\mathcal{L}(y_2) = 0$ we have $y_1 + \wp(y_2) = y_1$, hence $\wp(y_2) = 0$, i.e., 
$y_2 = 0$ and $\mathcal{L}(0) = 0$. Setting $y_1 = 0$, we then infer that 
$\mathcal{L}(\wp(y_2)) = \wp(\mathcal{L}(y_2))$ for every $y_2$. Consequently, 
$\mathcal{L}(y_1 + \wp(y_2)) = \mathcal{L}(y_1) + \wp(\mathcal{L}(y_2)) = \mathcal{L}(y_1) + \mathcal{L}(\wp(y_2))$
for all $y_1,y_2\in\F_{p^m}$. Since we can write every $y_3$ as $\wp(y_2)$, we infer that 
$\mathcal{L}(y_1+y_3) = \mathcal{L}(y_1) + \mathcal{L}(y_3)$ for all $y_1,y_3\in\F_{p^m}$.
\end{proof}

With Lemma \ref{equilem} we can completely describe $\wp$-affine (respectively linear) equivalence.
\begin{theorem}
$F_1,F_2:\F_{p^n}\rightarrow\F_{p^m}$ are $\wp$-affine equivalent if and only if $F_2(x) = \mathcal{A}_2(F_1(\mathcal{A}_1(x)))$, where $\mathcal{A}_i=\mathcal{L}_i+\alpha_i$, $\alpha_1\in\F_{p^n}$, $\alpha_2\in\F_{p^m}$, 
 $\mathcal{L}_1$  a linearized permutation of $\F_{p^n}$ and  $\mathcal{L}_2$ a linearized permutation of $\F_{p^m}$ such that 
$\mathcal{L}_2(\wp(x)) = \wp(\mathcal{L}_2(x))$ for all $x\in\F_{p^m}$.
In particular, $F_2(x) = \mathcal{A}_2(F_1(\mathcal{A}_1(x)))$ is P$\wp$N if and only if $F_1$ is P$\wp$N.
\end{theorem}
We observe that $\wp$-affine equivalence is included in conventional affine equivalence.
\begin{example}
Let $\wp(x) = cx$, then the condition $\wp(\mathcal{L}(x)) = \mathcal{L}(\wp(x))$ reduces to $c\mathcal{L}(x) = \mathcal{L}(cx)$.
If $c$ is not in the prime field $\F_p$, this is a restriction on the linear permutation $\mathcal{L}$.
\end{example}

Recall that for every orthomorphism $\wp$, every linear permutation (respectively linear balanced function) is P$\wp$N.
Clearly, two linear permutations are $\wp$-linear equivalent. Consequently, the linear permutations form one 
$\wp$-linear equivalence class of P$\wp$N functions on $\F_{p^n}$. 
As representative of the equivalence class,
we may choose $F(x) = x$. The quasigroup difference set corresponding to a linear permutation
is hence equivalent as a subset of $(\F_{p^n}\times\F_{p^m},+_\wp)$ to
$\{(x,x),\,:\,x\in\F_{p^n}\}$. 
Note that a similar statement can be made for affine permutations and equivalence under the condition $\wp$ linearized.

On the other hand, two quadratic permutations as P$\wp$N functions, if we just restrict to $\wp(x) = cx$, $c\in\F_p$, $c\ne 1$, are in general
not $\wp$-affine equivalent.

\section{Perspectives for future research}
\label{concl}

In this paper, for the first time, we find a connection between the $c$-differential uniformity (cDU) and combinatorial designs. In particular, we show that the graph of a PcN function   corresponds to a difference set in a quasigroup. Difference sets give rise to symmetric designs, which are known to construct optimal self complementary codes. Some types of designs can be also used in secret sharing and visual cryptography. We extend the PcN function to any orthomorphism $\wp$, not only $x\mapsto cx$, and that enables us to define an equivalence relation among  perfect $\wp$-nonlinear functions. We also provide an idea for a possible extension of the differential attack.

Lately there has been considerable progress in constructing PcN functions, in particular, but not only, in characteristic two. We refer
to the table in the appendix and the corresponding references. In this article we give two examples of P$\wp$N functions for $\wp(x)$ other than $cx$.
It is to be expected that many more classes of PcN ($c\neq 1$; recall that for $c=0$, PcN functions are simply permutations), and more general, P$\wp$N functions can be found with moderate effort, though, perhaps not in the binary case, where there is only {\em one} known nontrivial monomial PcN class (some only for $c=-1$), and about nine polynomials ones (constructed via some switching of a linearized polynomial).
Thus, in odd characteristic, this is somewhat opposite to the situation for planar functions.
In view of the above, at this point, other more general questions should be asked.

To give a complete description of all P$\wp$N functions for some given orthomorphism $\wp$, and in this way to describe (up to equivalence) all
difference sets of this type in the corresponding quasigroup, may be interesting.
Another interesting question is also whether non-linearized P$\wp$N functions exist for all orthomorphisms $\wp$, or whether for some orthomorphisms,
the corresponding quasigroup has only $\{(x,x)\,:\,x\in \F_{p^n}\}$ as a difference set (arising from a graph of a function).

Some questions may arise from the connection to permutation polynomials. Are there interesting permutation polynomials among $_\wp D_a(F(x))$  (equivalence questions would have to be addressed)? Are there (other) properties, which are specific to the permutations $_\wp D_a(F(x))$?

We conclude this article with an observation about a higher order $c$-differential attack, that may help circumvent the key addition non-cancellation in an extension of the differential attack.

We consider a round function ($S$-box) $F$  of a cipher (operating over a finite field of any characteristic~$p$) with a post-whitening key $K_1$. Computing the $c$-differential of $F+K_1$ at $c_1$, we get 
\[
{_{c_1}}D_a(F+K_1)(x)=F(x+a)+K_1-c_1\left( F(x)+K_1\right)=:G(x).
\]
As in the case of higher order differential cryptanalysis,  we continue with another round key $K_2$ and  obtain
\begin{align*}
{_{c_2}}D_b \left(G+K_2\right)(x)
&=G(x+b)-c_2G(x)+(1-c_2)K_2\\
&=F(x+a+b)-c_1 F(x+b)-c_2F(x+a)+c_1c_2 F(x)\\
&\qquad +(1-c_2)(1-c_1) K_1+(1-c_2) K_2.
\end{align*}
Thus if either $c_2=1$ (hence, the second derivative is the classical one), or $c_2\neq 1$ and the round key constants are related by $K_2=-(1-c_1)K_1$, the round keys will disappear and we get
\begin{align*}
{_{c_2}}D_b \left({_{c_1}}D_a\left(F+K_1\right)+K_2\right)(x)&={_{c_1}}D_aF(x+b)-c_2\cdot{_{c_1}}D_aF(x)\\
&={_{c_2}}D_b \left({_{c_1}}D_a\left(F\right)\right)(x).
\end{align*}
What that means is that the keyspace has to avoid round keys, like $K_1,K_2$, whose quotients $1+K_2/K_1$ cannot be a constant $c_1$ such that the  $c_1$-differential uniformity  of $F$ is rather high. In~\cite{ICSP21} it was shown that the second order $c$-differential uniformity with respect to $c$ (that is, $c_1=c_2=c$) is at least the value of the $c$-differential uniformity of~$F$.
Perhaps, it is worth investigating some of the known good cryptographic functions with respect to a sequence of derivatives, and investigate their higher order $c$-differential uniformity, as in~\cite{ICSP21}, since, as we see above, there are instances where the key addition disappears.

\section*{Acknowledgement}

N. A. and T.K. are supported by T\"UB\.ITAK Project under Grant 120F309.
W.M. is supported by the FWF Project P 35138. 
C.R. is supported by Research Council of Norway under Grants 311646. P.S. is partially supported by a grant from the NPS Foundation.

\section*{Appendix A}
\label{appendixA}

\begin{proposition}
Proposition \textup{\ref{NixPcN}} applies for all $n > 2(j+k)$.
\end{proposition}
\begin{proof}
To show the statement on a sufficient condition on the size of $n$, we can study geometric properties of the curve given in 
Equation \eqref{eq:z}, namely $z^{2^j-1}=c^{-1}\frac{y^{2^k}+ y^{2^k-1}+1}{y^{2^j(2^k+1)-1}}$, and its function field.
Let $F$ be the function field defined by Equation \eqref{eq:z}. As we observed, $F/\F_{2^n}(y)$ is a Kummer extension of degree $2^j-1$. Ramified places are determined by the zeros and the poles of $({y^{2^k}+ y^{2^k-1}+1})/{y^{2^j(2^k+1)-1}}$ and their multiplicities. Since the zeros of $y^{2^k}+ y^{2^k-1}+1$ are simple, they are totally ramified. Moreover, $\mathrm{gcd}(2^j-1,2^j(2^k+1)-1)=1$ implies that the zero of $y$ is totally ramified. Note that the multiplicity of the pole of $y$ is $(2^j-1)(2^k+1)$, i.e., it is divisible by the degree of the extension. Hence, the pole of $y$ is not ramified. 
Then by the Hurwitz genus formula, the genus $g(F)$ satisfies 
\begin{align*}
2g(F)-2=(2^j-1)(-2)+(2^j-2)(2^k+1), 
\end{align*}
i.e., $2g(F)=(2^j-2)(2^k-1)$. Hence the number $N(F)$ of rational places of $F$ satisfies 
\begin{align}\label{eq: N}
N(F) \geq 2^n+1-(2^j-2)(2^k-1)2^{n/2}.
\end{align}
Now we investigate the geometric properties of the curve $\mathcal{X}$ defined by Equation $ \eqref{eq:z}$, i.e., 
$f(y,z)=cz^{2^j-1}y^{2^j(2^k+1)-1}+y^{2^k}+ y^{2^k-1}+1$. There are two rational points of $\mathcal{X}$ lying at infinity, namely $(0:1:0)$ of multiplicity $2^j(2^k+1)-1$ corresponding to the unique rational place lying above the zero of $y$ and $(1:0:0)$ of multiplicity $2^j-1$ corresponding to the places lying over the pole of $y$. Hence, there are at most $2^j$ rational places corresponding to the points at infinity. Recall that an affine point $(\alpha,\beta)$ is a singular point of $\mathcal{X}$ if and only if $f(\alpha,\beta)=(\partial f(y,z)/ \partial y) (\alpha,\beta))
=(\partial f(y,z)/ \partial z )(\alpha,\beta))=0$, where $\partial f(y,z)/ \partial y$ and $\partial f(y,z)/ \partial z$ are the partial derivatives of $f$
with respect to $y$ and $z$, respectively. Since $\partial f(y,z)/ \partial y=cz^{2^j-1}y^{2^j(2^k+1)-2}+y^{2^k-2}$ and 
$\partial f(y,z)/ \partial z=cz^{2^j-2}y^{2^j(2^k+1)-1}$, the curve $\mathcal{X}$ has no affine singular points. It is a well-known fact that each non-singular rational point corresponds to a unique rational place. From the above argument and Equation \eqref{eq: N}, we conclude that the number $N(\mathcal{X})$ of rational affine points of $\mathcal{X}$ satisfies 
\begin{align*}
N(\mathcal{X})\geq 2^n-(2^j-2)(2^k-1)2^{n/2}-(2^j-1).
\end{align*}
Moreover, the line defined by $y$ intersects $\mathcal{X}$ only at infinity, and the line defined by $z$ intersects $\mathcal{X}$ at most at $2^k$ affine rational points. 
Hence, the number $N$ of rational points $(y,z)$ with $yz\neq 0$ satisfies $N\geq 2^n-(2^j-2)(2^k-1)2^{n/2}-(2^j-1)-2^k$, which gives the desired result.
\end{proof}

\section*{Appendix B}
\label{appendixB}

 We include here two tables, which are taken from~\cite{LRS22} and updated, containing some of the known classes with low $c$-differential uniformity (cDU) (we make the choice to include only the ones whose cDU is less than~4, unless it is a very known function, or is another case of a function with low cDU). We note that over the binary fields, there are not too many classes of PcN functions.
 
 We use $v_2$ as the $2$-valuation of the input, that is the largest power of $2$ dividing the input; the inverse is taken in the sense of modulo $p^n-1$ for the respective prime~$p$.
 Table~\ref{tab:long} lists the exponent of some monomials $x^d$. Table~\ref{tab2:long}  lists the known polynomials with low $c$-differential uniformity (here, $l > 1$ is a divisor of $p^n - 1$ and $g$ is a primitive element of $\F_{p^n}$, and $D_0$ is the multiplicative subgroup of $\F_{p^n}$ generated by $g$).
  \begin{center}
  \footnotesize{
   \begin{longtable}{|p{3cm}|c|p{2cm}|p{4.2cm}|c|}
   \caption{${_c}\Delta_F$ of various classes of functions $x^d$, $c \neq 1$} \label{tab:long} \\
   
   \hline \multicolumn{1}{|p{3cm}|}{$d$} & \multicolumn{1}{c|}{$\F_{p^n}$} & \multicolumn{1}{c|}{${_c}\Delta_F$} &  \multicolumn{1}{p{4.2cm}|}{Conditions} & \multicolumn{1}{c|}{Ref}\\ \hline 
\endfirsthead
   
   \multicolumn{5}{c}%
{{\tablename\ \thetable{} -- continued from previous page}} \\
   
 \hline \multicolumn{1}{|p{3cm}|}{$d$} & \multicolumn{1}{c|}{$\F_{p^n}$} & \multicolumn{1}{c|}{${_c}\Delta_F$} &  \multicolumn{1}{p{4.2cm}|}{Conditions} & \multicolumn{1}{c|}{Ref}\\ \hline 
\endhead
   
\hline \multicolumn{5}{|r|}{{Continued on next page}} \\ \hline
\endfoot   

\hline
\endlastfoot

     $2$ & $p>2$  & 2 (AP$c$N) & none &\cite{efrst}\\
              
       \hline
        $\frac{3^k+1}{2}$ & $p=3$ & 1 (P$c$N)   & $c=-1$, $\frac{2n}{\gcd(k,2n)}$ is odd&\cite{efrst} \\        
       \hline
       ${p^n-2}$ & any $p$ & 1 (P$c$N)    & $c=0$ &\cite{efrst} \\        
       \hline
         ${2^n-2}$ & $p=2$ & 2 (AP$c$N)    & $c \neq 0$, $\Trn(c)=\Trn(1/c)=1$ &\cite{efrst} \\        
       \hline
        ${2^n-2}$ & $p=2$ & 3  & $c \neq 0$, $\Trn(c)=0$ or $\Trn(1/c)=0$ &\cite{efrst} \\        
       \hline
 	${p^n-2}$ & $p>2$ & 2 (AP$c$N)    & $c \neq 0$, $(c^2-4c) \notin [\F_{p^n}]^2$, $(1-4c) \notin [\F_{p^n}]^2$, or $c=4,4^{-1}$ &\cite{efrst} \\        
       \hline
       ${p^n-2}$ & $p>2$ & 3   & $c\neq 0,4,4^{-1}$, $(c^2-4c) \in [\F_{p^n}]^2$ or $(1-4c) \in [\F_{p^n}]^2$ &\cite{efrst} \\        
       \hline
         ${2^k+1}$ & $p=2$ & $\frac{2^{\gcd(2k,n)}-1}{2^{\gcd(k,n)}-1}$ & $c \in \F_{2^{\gcd(n,k)}}\setminus \{1\}$, $\frac{n}{\gcd(n,k)} \geq 3 (n\geq3)$ &\cite{MRSZY21} \\   
         \hline
         ${2^k+1}$ & $p=2$ & $2^{\gcd(n,k)}+1$ & $c \in \F_{2^n}\setminus \F_{2^{\gcd(n,k)}}$ &\cite{MRSZY21} \\   
       \hline
   	${p^k+1}$ & any $p$ & $  \gcd(p^k+1,p^n-1)$ & $c \in \F_{p^{\gcd(n,k)}}$&\cite{MRSZY21} \\        
       \hline
        $\frac{p^k+1}{2}$ & $p>2$ &  $p^{\gcd(n,k)}+1$  & $c =-1$ &\cite{MRSZY21}\\        
       \hline
        $\frac{p^n+1}{2}$ & $p>2$ & $\leq 4$ & $c\neq \pm 1$&\cite{MRSZY21} \\  
        \hline
        $\frac{p^n+1}{2}$ & $p>2$ & $\leq 2$ & $c\neq \pm 1$, $\eta\big(\frac{1-c}{1+c}\big)=1$ $p^n\equiv 1 $ (mod 4)&\cite{MRSZY21} \\ 
       \hline
  	$\frac{2p^n-1}{3}$ & any & $\leq 3$   & $p^n \equiv 2\pmod 3$ &\cite{MRSZY21}\\  
       \hline
       $\frac{p^n+3}{2}$ & $p>3$ & $\leq 3$   & $c=-1$, $p^n \equiv 3\pmod 4$ &\cite{MRSZY21}\\   
       \hline
      $\frac{p^n+3}{2}$ & $p>3$ & $\leq 4$   & $c=-1$, $p^n \equiv 1\pmod 4$ &\cite{MRSZY21}\\  
       \hline
       $\frac{p^n-3}{2}$ & $p>2$ & $\leq 4$   & $c=-1$ &\cite{MRSZY21}\\        
       \hline
	$\frac{3^n+3}{2}$ & $p=3$ & 2 (AP$c$N)    & $c=-1$, $n$ even &\cite{MRSZY21}\\   
	\hline
	$\frac{3^n-3}{2}$ & $p=3$ & 6  & $c=-1$, $n=0\pmod 4$ &\cite{MRSZY21}\\ 
	\hline
	$\frac{3^n-3}{2}$ & $p=3$ & 4  & $c=-1$, $n\neq0\pmod 4$ &\cite{MRSZY21}\\ 
	\hline
	$\frac{3^n-3}{2}$ & $p=3$ & 2 (AP$c$N)  & $c=0$, &\cite{MRSZY21}\\ 
       \hline
       $\frac{3^n+1}{4}\, (\frac{3^k+1}{4})^{-1}$  & $p=3$ & 1 (P$c$N)   & $n,k$ odd, $c=-1$,  $\gcd(n,k)=1$  &\cite{ZH21}\\        
       \hline
         $\frac{5^n-1}{2} + (\frac{5^k+1}{2})^{-1} $ & $p=5$ & 1 (P$c$N)   & $n,k$ odd, $c=-1$,  $\gcd(n,k)=1$  &\cite{ZH21}\\        
       \hline
         $\frac{p^n+1}{2}\,(p^k+1)^{-1}$ & $p>2$ & $\leq 6$  & $d$ even, $c=-1$, $p^n \equiv 3 \pmod 4$  &\cite{ZH21}\\        
       \hline
       $\frac{p^n+1}{2}\, (p^k+1)^{-1}$ & $p>2$ & $\leq 3$  & $d$ odd, $c=-1$,  $p^n \equiv 3 \pmod 4$  &\cite{ZH21}\\        
       \hline
         $\frac{p^n+1}{4}+\frac{p^n-1}{2}$ & $p>2$ & $\leq 3$    & $c=-1$, $p^n \equiv 7 \pmod{8}$  &\cite{ZH21}\\ 
         	\hline
	${\frac{p^n-1}{2}+p^k+1}$ & $p>2$ & $\leq 3$    & $c=-1$, $\frac{n}{\gcd(n,k)}$ odd, $p^n \equiv 3 \pmod{4}$  &\cite{ZH21}\\ 
       \hline
	${\frac{p^n-1}{2}+p^k+1}$ & $p>2$ & $\leq 6$    & $c=-1$, $\frac{n}{\gcd(n,k)}$ odd, $p^n \equiv 1 \pmod{4}$  &\cite{ZH21}\\  

	 \hline
	${\frac{p^l+1}{2}}$ &  $p>2$ & 1 (P$c$N)    & $c=-1$, $l=0$ or $l$ even and $n$ odd, or $l,n$ both even together with $t_2 \geq t_1 +1$, where $n=2^{t_1u}$ and $l=2^{t_2}$ such that $2 \not |  u,v$&\cite{HPRS21}\\ 
       \hline
	${\frac{p^l+1}{2}}$ & $p>2$ & $\frac{p+1}{2}$    & $c=-1$, $\gcd(l,2n)=1$, $p\equiv 1\pmod{4}$ or $p\equiv 3 \pmod{8}$ &\cite{HPRS21}\\        
	 \hline
	${\frac{5^l+1}{2}}$ & $p=5$ & 3 & $c=-1$, $\gcd(l,2n)=1$ &\cite{HPRS21}\\        
	 \hline
	${\frac{3^l+1}{2}}$ & $p=3$ & 2 (AP$c$N) & $c=-1$, $\gcd(l,2n)=1$ &\cite{HPRS21}\\     
\hline
$p^4+(p-2)p^2$ +  $p(p-1)+1$ & $p>2$ & 1 (P$c$N)    & $c=-1$, $n=5$  &\cite{HPRS21}\\        
\hline
	${\frac{p^5+1}{p+1}}$ & $p>2$ & 1 (P$c$N) & $c=-1$, $n=5$ &\cite{HPRS21}\\
\hline
$(p-1)p^6+p^5+(p-2)p^3+(p-1)p^2+p$ & $p>2$ & 1 (P$c$N)    & $c=-1$, $n=7$  &\cite{HPRS21}\\
\hline
	${\frac{p^7+1}{p+1}}$ & $p>2$ & 1 (P$c$N) & $c=-1$, $n=7$ &\cite{HPRS21}\\
\hline
${\frac{3^n+7}{2}}$ & $p=3$ & $\leq 2$ (AP$c$N) & $c=-1$, $n$ odd  &\cite{WLZ21}\\  
\hline
	$\frac{3^{\frac{n+1}{2}-1}}{2}$ & $p=3$ & $\leq 2$ (AP$c$N) & $c=-1$, $n\equiv 1 \pmod{4}$ &\cite{Yan21}\\ 
\hline
$\frac{3^{\frac{n+1}{2}-1}}{2}+\frac{3^n-1}{2}$ & $p=3$ & $\leq 2$ (AP$c$N) & $c=-1$, $n\equiv 3 \pmod{4}$ &\cite{Yan21}\\ 
\hline
$\frac{3^{n+1}-1}{8}$ & $p=3$ &  $\leq 2$ (AP$c$N) & $c=-1$, $n\equiv 1 \pmod{4}$ &\cite{Yan21}\\ 
\hline
$\frac{3^{n+1}-1}{8}+\frac{3^n-1}{2}$ & $p=3$ &  $\leq 2$ (AP$c$N) & $c=-1$, $n\equiv 3 \pmod{4}$ &\cite{Yan21}\\ 
\hline
$(3^{\frac{n+1}{4}}-1)(3^{\frac{n+1}{2}}+1)$ & $p=3$ &  $\leq 4$ & $c=-1$, $n\equiv 3 \pmod{4}$ &\cite{Yan21}\\ 
\hline
$\frac{3^n+1}{4}+\frac{3^n-1}{2}$ & $p=3$ &  $\leq 4$ & $c=-1$, $n$ odd &\cite{Yan21}\\ 
\hline
${d^{-1}}\pmod{p^n-1}$ & any $p$ & 1 (P$c'$N)  & $x^d$ is P$c$N, $c'=c^d$ &\cite{WZ21}\\ 
\hline
$\{2^j,2^j(2^k+1),k,j\geq 0\}$ & $p=2$ & 1 (P$c$N)  &  &\cite{WZ21}\\ 
\hline
odd $2 (p^k+1)^{-1}\pmod{p^n-1},k\geq 0$ & $p>2$ & 1 (P$c$N)  & $c=-1$ &\cite{WZ21}\\
\hline
$\frac{p^n+1}{2}\left(\frac{p^k+1}{2} \right)^{-1}  $ & $p>2$ & 1 (P$c$N)  & $c=-1$, $v_2(k)=v_2(n)$, \hspace{8mm} $p^n\equiv 1 \pmod{4}$ &\cite{WZ21}\\
\hline
         \end{longtable}
        }
            \end{center}
        
         \begin{center}
        	\scriptsize{
        		\begin{longtable}{|p{5cm}|c|p{1.55cm}|p{3cm}|c|}
        			\caption{${_c}\Delta_F$ of various classes of functions $F(x)$, $c \neq 1$} \label{tab2:long} \\
        			
        			\hline \multicolumn{1}{|p{5.2cm}|}{$F(x)$} & \multicolumn{1}{c|}{$\F_{p^n}$} & \multicolumn{1}{c|}{${_c}\Delta_F$} &  \multicolumn{1}{p{3cm}|}{Conditions} & \multicolumn{1}{c|}{Ref}\\ \hline 
        			\endfirsthead
        			
        			\multicolumn{5}{c}%
        			{{\tablename\ \thetable{} -- continued from previous page}} \\
        			
        			\hline \multicolumn{1}{|p{5.2cm}|}{$F(x)$} & \multicolumn{1}{c|}{$\F_{p^n}$} & \multicolumn{1}{c|}{${_c}\Delta_F$} &  \multicolumn{1}{p{3cm}|}{Conditions} & \multicolumn{1}{c|}{Ref}\\ \hline 
        			\endhead
        			
        			\hline \multicolumn{5}{|r|}{{Continued on next page}} \\ \hline
        			\endfoot   
        			
        			\hline
        			\endlastfoot
        			
        			   $x^{10} - ux^6 - u^2x^2$ & $p=3$ & $\geq 2$ & $u\in \F_{3^n}$&\cite{efrst} \\ 
        			\hline
        			$L(x)(\sum_{i=1}^{l-1} L(x)^{\frac{p^n-1}{l}i}+u)$ & any $p$ & $\leq 2$ (AP$c$N) & $L$ an $\F_p$-linearized polynomial,  $l|(p^n-1)$, $u \neq 1,(1-l) \mod p$, $1-\frac{l}{(1-c)(u+l-1)},1+\frac{l}{(1-c)(u-1)} \in D_0$ &\cite{WLZ21}\\  
        			\hline
        			$(x^{p^k}-x)^{\frac{q-1}{2}+1}+a_1x+a_2x^{p^k}+a_3x^{p^{2k}}$ & $p=3$ & $\leq 2$ (AP$c$N) & $c=-1$, $0\leq i\leq2$, $a_1,a_2,a_3\in \F_3$, $a_1+a_2+a_3\neq 0$  &\cite{WLZ21}\\  
        			\hline    			
        			$f(x)(\Trn(x)+1)+f(x+\gamma)\Trn(x)$ & $p=2$ & 1 (P$c$N) &  $f(x)$ is P$c$N, $\gamma \in \F_{p^n}^*$  &\cite{WLZ21}\\  
        			\hline
        			$L(x)+L(\gamma)(\Trn(x))^{q-1}$ &any $p$ & 1 (P$c$N) & $L$ an $\F_q$-linearized polynomial, $\gamma \in \F_{q}^*$, $\Trn(\gamma) =0$  &\cite{WLZ21}\\  
        			\hline
        			$u\phi(x)+g((\Trn(x))^q)-g(\Trn(x))$ &any $p$ & 1 (P$c$N) & $\phi$ an $\F_q$-linearized polynomial, $u\in \F_q^*$, ker($\phi$)$\cap$ ker($\Trn$)=$\{0\}$, $g \in \F_{q^n}[x]$  &\cite{WLZ21}\\ 
        			\hline
        			$u(x^q-x)+g(\Trn(x))$ &any $p$ & 1 (P$c$N) & $g \in \F_{q^n}[x]$ a permutation of $\F_q$, $u \in \F_{q}^*$, $p \nmid n$  &\cite{WLZ21}\\ 
        			\hline
        			$F(x)+u\Trn(vF(x))$ & any $p$ & 1 (P$c$N)  & $F$ is P$c$N, $\Trn(-uv)\neq 1$ & \cite{LRS22}\\
        			\hline
        			$L_1(x)+L_1(\gamma)\Trn(L_2(x))$ & any $p$  &  1 (P$c$N)   &   $\Trn\left(\frac{L_1(\gamma)}{1-c} \right)=0$, $\Trn(\gamma)=0$ & \cite{LRS22}\\
        			\hline
        			$L(x)+\prod_{i=1}^s\left(\Trn(x^{2^{k_i}+1}+\delta_i)\right)^{g_i}$ & $p=2$ &  $\leq 2$ (AP$c$N) & $1\leq k_i \leq n-1$ & \cite{LRS22}\\
        			\hline
        			$L(x)+\prod_{i=1}^s\left(\alpha_i\Tr_{q^n/q^m}(x^{2^{k_i}+1}+\delta_i)\right)^{g_i}$ & $p=2$ &  $\leq 2$ (AP$c$N) & $g_i\geq 1,\,\delta_i\in\F_{2^n}$, $\alpha_i\in\F_{2^m}^*$, $1\leq k_i \leq n-1$ & \cite{LRS22}\\
        			\hline
 $L(x)+u\sum_{i=1}^t\left(\Tr_{q^n/q^m}(x)^{k_i}+\delta_i \right)^{s_i}$ & any $p$ & 1 (P$c$N) & 
        			$pm\,|\,n$, $1\leq t\in\mathbb{Z}_{>0}$, $u\in\F_{p^m}^*$, $\delta_i\in\F_{p^m}$, 
        			$1\leq k_i,s_i\leq p^n-1$, $L$ linearized permutation, $c\in \F_{p^m}\setminus\{1\}$ & \cite{LRS22}\\
        			\hline
$(x^{2^m}+x+\delta)^{2^{2m}+1}+x$ & $p=2$ & 1 (PcN) & $n=3m$, $c\in \F_{2^m}\setminus\{1\}$, $\delta\in\F_{2^n}$ & \cite{GHS22}\\
\hline
$(x^{2^m}+x+\delta)^{2^{2m}+1}+x$ & $p=2$ & 2 (APcN) & $n=3m$, $c\in \F_{2^n}\setminus \F_{2^m}$, ${\rm Tr}_m^{3m}(\delta)=1$ & \cite{GHS22}\\
\hline
$(x^{2^m}+x+\delta)^{2^{2m}+1}+x$ & $p=2$ & $\leq 4$ & $n=3m$, $c\in \F_{2^n}\setminus \F_{2^m}$, ${\rm Tr}_m^{3m}(\delta)\neq 1$ & \cite{GHS22}\\
\hline
$(x^{2^m}+x+\delta)^{2^{2m-1}+2^{m-1}}+x$ & $p=2$ & 1 (PcN) &$n=3m,m\not\equiv\pm1\pmod 3$, $c\in \F_{2^m}\setminus\{1\}$, $\delta\in\F_{2^n}$ & \cite{GHS22}\\
\hline
$(x^{2^m}+x+\delta)^{2^{2m-1}+2^{m-1}}+x$  & $p=2$ & 2 (APcN) & $n=3m,m\not\equiv\pm1\pmod 3$, $c\in \F_{2^n}\setminus \F_{2^m}$, ${\rm Tr}_m^{3m}(\delta)=0$ & \cite{GHS22}\\
\hline
$(x^{2^m}+x+\delta)^{2^{2m-1}+2^{m-1}}+x$  & $p=2$ & $\leq 4$ & $n=3m,m\not\equiv\pm1\pmod 3$, $c\in \F_{2^n}\setminus \F_{2^m}$, ${\rm Tr}_m^{3m}(\delta)\neq 0$ & \cite{GHS22}\\
\hline
$(x^{3^m}-x+\delta)^{3^{2^m-1}+2\cdot3^{m-1}}+x$ & $p=3$ & 1 (PcN) &$n=2m$, $c\in \F_{3^m}\setminus\{1\}$, $\delta\in\F_{2^n}$, or  $c\in \F_{3^n}\setminus \F_{3^m}$, ${\rm Tr}_m^{2m}(\delta)=0$ & \cite{GHS22}\\
\hline
$(x^{3^m}-x+\delta)^{3^{2^m-1}+2\cdot3^{m-1}}+x$  & $p=3$ & $3$ & $n=2m$, $c\notin \F_{3^m}$, ${\rm Tr}_m^{2m}(\delta)\neq 0$ & \cite{GHS22}\\
\hline
$(x^{p^m}-x+\delta)^{p^{m+1}+1}+x$ &  $p>2$ & 1 (PcN) &$n=2m$, $c\in \F_{p^m}\setminus\{1\}$, ${\rm Tr}_m^{2m}(\delta)=0$, or $\frac{{\rm Tr}_m^{2m}(\delta)-1}{{\rm Tr}_m^{2m}(\delta)}$ is a $(p-1)$-th power & \cite{GHS22}\\
\hline
$(x^{p^m}-x+\delta)^{p^{m+1}+1}+x$ &  $p>2$ & $p$ &$n=2m$, $c\in \F_{p^n}\setminus \F_{p^m}$, ${\rm Tr}_m^{2m}(\delta)=0$, or $\frac{{\rm Tr}_m^{2m}(\delta)-1}{{\rm Tr}_m^{2m}(\delta)}$ is a $(p-1)$-th power & \cite{GHS22}\\
\hline
  \end{longtable}
}
\end{center}

\end{document}